\documentclass[12pt, draftclsnofoot, onecolumn]{IEEEtran}

\usepackage{amsfonts}
\usepackage{slashbox}
\usepackage{amsmath}
\usepackage{multirow}
\usepackage{graphicx}
\usepackage{amsfonts}
\usepackage{amsmath,epsfig}
\usepackage{mathbbold}
\usepackage{stfloats}
\usepackage{amssymb}
\usepackage{url}
\usepackage{bm}
\usepackage{cite}
\usepackage{amsmath}
\usepackage{amssymb}
\usepackage{latexsym}
\usepackage{multirow}
\usepackage{epsfig}
\usepackage{graphics}
\usepackage{mathrsfs}
\usepackage{algorithmic}
\usepackage{algorithm}
\usepackage{subfigure}
\usepackage{slashbox}
\usepackage[all]{xy}

\usepackage{pifont}
\usepackage{bbding}
\usepackage{stmaryrd}
\usepackage{amssymb}
\usepackage{amsfonts}
\usepackage{epic}
\usepackage{stfloats}
\usepackage{epstopdf}
\usepackage{bm}
\usepackage{xcolor}
\usepackage{mathrsfs}
\usepackage{pifont}
\usepackage{bbding}
\usepackage{amsmath,epsfig}
\usepackage{mathbbold}
\usepackage{stmaryrd}
\usepackage{amsmath}
\usepackage{amssymb}
\usepackage{amsthm}

\usepackage{amsfonts}
\usepackage{epic}
\usepackage{graphicx}
\usepackage{subfigure}
\usepackage{enumerate}
\usepackage{latexsym}
\usepackage{epstopdf}
\usepackage{multirow}
\usepackage{stfloats}
\usepackage{bm}
\usepackage{booktabs}
\usepackage{color}
\usepackage{cases}
\usepackage{array}
\usepackage{makecell}
\usepackage{times}

\newtheorem{remark}{Remark}
\newtheorem{proposition}{Proposition}
\newtheorem{definition}{Definition}

\newcommand{\A}{\bm A}
\newcommand{\w}{\bm w}

\newcommand{\Ss}{\bm S}
\newcommand{\G}{\bm G}
\newcommand{\HH}{\bm H}

\newcommand{\I}{\bm{I}}

\newcommand{\W}{\bm W}
\newcommand{\PP}{\bm P}

\newcommand{\B}{\bm B}

\newcommand{\x}{ \bm{x} }
\newcommand{\y}{ \bm{y} }

\newcommand{\cc}{\mathbf c}
\newcommand{\s}{\mathbf s}
\newcommand{\ttheta}{\mathbf \Theta}

\newcommand{\hh}{\bm h}
\newcommand{\g}{\bm g}

\graphicspath{{./figs/}}

\begin{document}
\!\!\!\title{Beamforming Optimization for Wireless Network Aided by  Intelligent Reflecting   Surface with Discrete  Phase Shifts }
\author{\IEEEauthorblockN{Qingqing Wu,  \emph{Member, IEEE} and Rui Zhang, \emph{Fellow, IEEE}
\thanks{ The authors are with the Department of Electrical and Computer Engineering, National University of Singapore, email:\{elewuqq, elezhang\}@nus.edu.sg. Part of this work has been  presented in  \cite{wu2018IRS_discrete}.}}  }

\maketitle
\vspace{-6mm}
\begin{abstract}
Intelligent reflecting surface (IRS) is  a cost-effective solution  for achieving high spectrum and energy efficiency in future wireless networks   by leveraging massive low-cost passive elements that are able to reflect the signals with adjustable phase shifts. Prior  works on IRS mainly consider continuous phase shifts at reflecting elements, which are practically difficult to implement due to the hardware limitation.  In contrast, we study in this paper an IRS-aided wireless network, where an IRS with only a finite number of phase shifts at each element  is deployed to assist in  the communication from a multi-antenna access point (AP) to multiple single-antenna users. We aim to minimize the transmit power at the AP by jointly optimizing the continuous  transmit precoding at the AP and the discrete reflect phase shifts  at the IRS, subject to a given set of minimum signal-to-interference-plus-noise ratio (SINR) constraints at the user receivers. The considered problem is shown to be a mixed-integer non-linear program (MINLP) and thus is difficult to solve in general. To tackle this problem, we first study the single-user case with one user assisted by the IRS and propose both optimal and suboptimal algorithms for solving it. Besides,  we analytically show that as compared to the ideal  case with continuous phase shifts,  the IRS with discrete phase shifts achieves the same squared power gain in terms of  asymptotically large number of reflecting elements, while a constant proportional power  loss is incurred that depends only on the number of phase-shift levels. The proposed designs for the single-user case are also  extended to the general setup with multiple users among which some are aided by the IRS.  Simulation results verify our performance analysis  as well as the effectiveness of our proposed designs as compared to various  benchmark schemes.
\end{abstract}
\vspace{-6mm}

\section{Introduction}
Although massive multiple-input multiple-output (MIMO) technology has significantly improved the spectrum efficiency of wireless communication systems, the required high complexity, high energy consumption, and high hardware cost are  still  the main hindrances to its implementation in practice  \cite{zhang2016fundamental,Hien2013,wu2016overview,foad16jstsp}.  Recently, intelligent reflecting surface (IRS) has been  proposed as a new and  cost-effective solution for achieving high  spectrum and energy efficiency for wireless communications  via only low-cost reflecting elements \cite{JR:wu2019IRSmaga,JR:wu2018IRS}. An IRS is generally composed of a large number of passive elements each being able to reflect the incident signal with an adjustable phase shift. By smartly tuning the phase shifts of all elements adaptively according  to the dynamic wireless channels,  the  signals reflected by an IRS can add constructively or destructively with those non-reflected by it at a nearby user receiver to boost the desired  signal power and/or suppress the co-channel interference, thus significantly  enhancing the communication performance without the need of deploying additional  active base stations (BSs) or relays. In addition, without employing any  transmit radio frequency (RF)  chains, IRSs usually have much smaller signal coverage  than active BSs/relays, which makes it easier to practically deploy them without interfering each other. Moreover,   from the implementation  perspective, IRSs possess appealing  features such as low profile and lightweight, thus can be easily mounted on walls or ceilings of buildings, while integrating them into the existing   cellular and WiFi systems does not require any change in the  hardware at the BSs/access points (APs) as well as  user terminals. As compared to existing wireless technologies based on active elements/RF chains such as MIMO relay and massive MIMO, it has been shown in \cite{JR:wu2019IRSmaga,JR:wu2018IRS} that IRS with only passive reflecting elements can potentially yield superior performance scaling with the increasing number of elements, but at substantially reduced hardware and energy costs. It is worth noting that there have been other terminologies similar to IRS proposed in the literature, such as intelligent wall \cite{subrt2012intelligent}, passive intelligent mirror \cite{huangachievable,huang2018largeRIS},  smart reflect-array \cite{tan2018enabling}, and reconfigurable metasurface \cite{di2019smart}, among others.

\begin{figure}[!t]
\centering
~~~~~~~~~~~~~\includegraphics[width=0.95\textwidth]{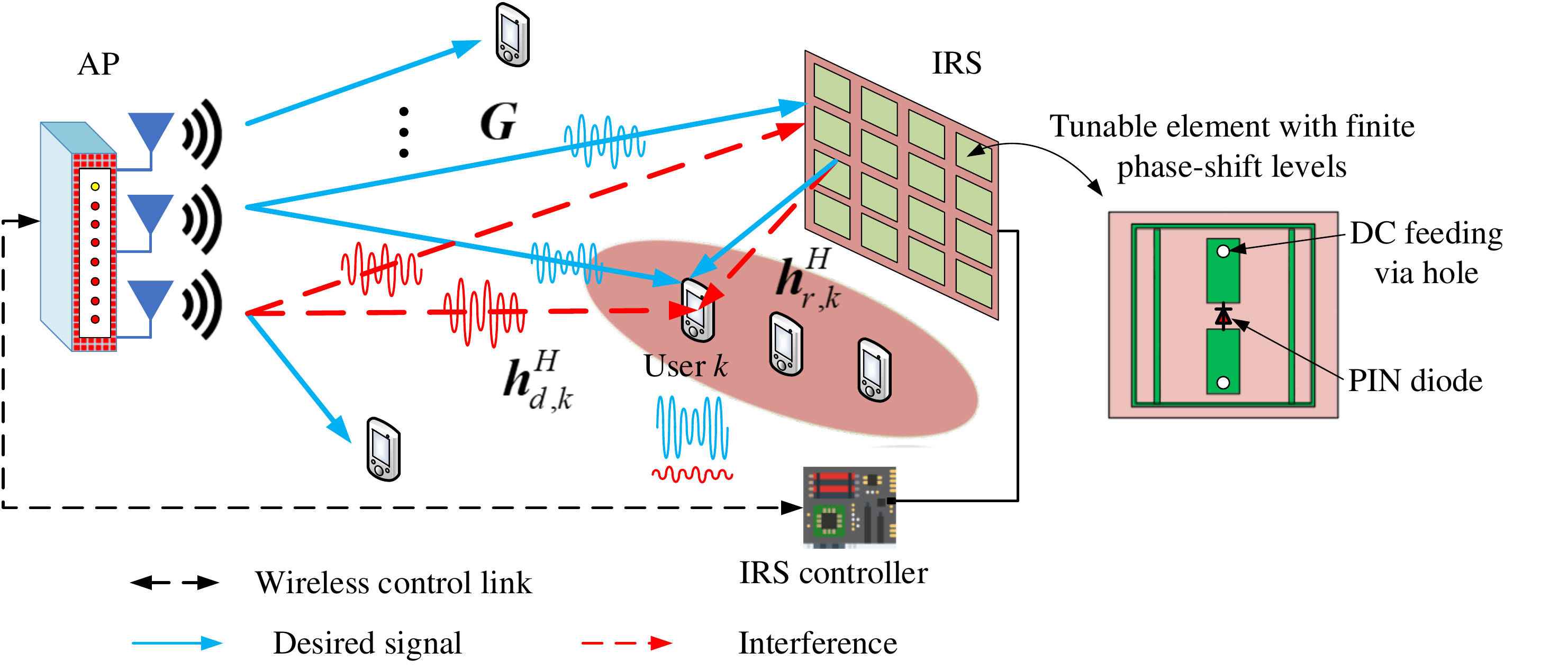}
\caption{An IRS-aided multiuser wireless communication system. } \label{system:model}
\end{figure}

IRS-aided wireless communications have drawn significant research  attention recently (see, e.g.  \cite{JR:wu2019IRSmaga,wu2018IRS,JR:wu2018IRS,jiang2019over,yan2019passive,yu2019miso,cui2019secure,yang2019irs}).
Specifically, for the IRS-aided wireless system with a single user,  it was shown in  \cite{wu2018IRS,JR:wu2018IRS} that the IRS is capable of creating a ``signal hot spot'' in its vicinity via joint active beamforming at the BS/AP and passive beamforming at the IRS. In particular, an asymptotic receive signal power or signal-to-noise ratio (SNR) gain in the order of  $\mathcal{O}(N^2)$, with $N$ denoting the number of reflecting elements at the IRS, was shown as  $N\rightarrow \infty$. Such a squared  power/SNR gain is larger than that of massive MIMO, i.e., $\mathcal{O}(N)$  \cite{Hien2013}, due to the fact that the IRS combines the functionalities of both receive and transmit arrays for energy harvesting and reflect beamforming, respectively,   thus doubling the gain as compared to massive MIMO with a constant total transmit power regardless of $N$. In addition, with conventional MIMO relays (even assuming their full-duplex operation with perfect self-interference cancellation), the  SNR at the user receiver increases with the number of active antennas, $N$, only with $\mathcal{O}(N)$ due to the noise effect at the relay  \cite{JR:wu2018IRS}, which is also lower than $\mathcal{O}(N^2)$ of the IRS thanks to its full-duplex and noise-free signal reflection.  Furthermore, for a general multiuser system aided by IRS as shown in Fig. 1,  it was shown in \cite{JR:wu2018IRS} that besides enhancing the desired signal power/SNR  at the user receiver,  a virtually ``interference-free'' zone can also be established in the proximity of the IRS, by exploiting  its spatial interference nulling/cancellation  capability.
 In particular,  a user near the IRS is  able to tolerate more  interference from the AP as compared to others outside the coverage region of  the IRS.
This thus provides more flexibility  for designing the transmit precoding at the AP for such ``outside'' users  and as a result improves the signal-to-interference-plus-noise ratio (SINR)  performance of all users in the system.

The hardware implementation of IRS is based on the concept of   ``metasurface'', which is made of  two-dimensional (2D) metamaterial with high controllability \cite{JR:wu2019IRSmaga,cui2014coding,kaina2014shaping,zhang2018space,di2019smart}.  By properly designing  each element of the metasurface, including geometry shape (e.g., square or split-ring), size/dimension,  orientation, etc.,  different phase shifts of the reflected signal can be resulted.    
However,  for wireless communication applications, it is more desirable to adjust the phase shifts by the IRS in real time to  cater for time-varying wireless channels arising from the user mobility. This can be realized  by leveraging electronic devices such as positive-intrinsic-negative (PIN) diodes,  micro-electromechanical system (MEMS) based switches, or field-effect transistors (FETs) \cite{nayeri2018reflectarray}.  In Fig. \ref{system:model}, we show one example of a tunable element's structure, in which a PIN  diode is embedded  in the  center to achieve the binary phase shifting. Specifically, the PIN diode can be switched between  ``On'' and ``Off'' states by controlling its biasing voltage via a direct-current (DC) feeding line, thereby generating  a phase-shift difference of $\pi$ in rad \cite{JR:wu2019IRSmaga,cui2014coding}. 
In practice, different phase shifts at IRS's elements can be realized independently via setting the corresponding  biasing voltages by using a smart IRS controller. Therefore, the main power consumption of IRS is due to the feeding circuit of the diode used to tune the elements (on the order of microwatts \cite{kaina2014shaping}),  which is significantly lower than that by transmit RF chains of conventional active arrays  \cite{Hien2013,cui2004energy}.

The existing works \cite{wu2018IRS,JR:wu2018IRS,jiang2019over,cui2019secure,yan2019passive,yu2019miso,huangachievable,huang2018largeRIS,yang2019irs} on IRS-aided wireless communications  are mainly based on the assumption of  continuous phase shifts at its reflecting elements. However,  in practice, this is difficult to realize since manufacturing each reflecting element with more levels of phase shifts incurs a higher cost, which may not be scalable as the number of elements for IRS is usually very  large. For example, to enable $16$ levels of phase shifts, $\log_216=4$ PIN diodes shown in Fig. 1  need to be integrated to each element. This not only makes the element design more challenging due to its limited size, but also requires extra controlling pins at the IRS controller to control more  PIN diodes. Although a single varactor diode can be used to achieve more than two phase shifts, it requires a wide range of biasing voltages.  As such,  for practical IRSs with  a large number of elements, it is more cost-effective to implement only discrete phase shifts with  a small number of control bits for each element, e.g., $1$-bit for  two-level ($0$ or $\pi$) phase shifts \cite{cui2014coding,wu2018IRS_discrete,tan2018enabling}.
Note that such limited discrete phase shifts inevitably cause misalignment of IRS-reflected and non-IRS-reflected  signals at designated  receivers and thus result in certain performance degradation, which needs to be investigated for practically  deploying  low-resolution IRSs in future wireless systems. Besides, it remains unknown whether the ``squared power/SNR  gain'' revealed in \cite{JR:wu2018IRS} still holds for the case of IRS with discrete phase shifts, as well as how the resultant discrete-phase constraints impact the IRS's passive beamforming design jointly with the AP's active transmit precoding.

Motivated by the above, we study  in this paper an IRS-aided multiuser wireless communication system shown in Fig. \ref{system:model}, where  a multi-antenna AP serves multiple single-antenna users with the help of an IRS. In contrast to the continuous phase shifts assumed in \cite{wu2018IRS,JR:wu2018IRS,jiang2019over,cui2019secure,yan2019passive,yu2019miso,huangachievable,huang2018largeRIS,yang2019irs}, we consider the practical case where each element of the IRS has only a finite number of discrete  phase shifts. We aim to minimize the transmit power required at the AP via jointly optimizing the  active transmit precoder at the AP and passive reflect discrete phase shifts  at the IRS, subject to a given set of SINR constraints at the user receivers. However,  the transmit precoder and  discrete phase shifts are intricately  coupled in the SINR constraints,  rendering the formulated optimization problem  a mixed-integer non-linear program (MINLP) that  is NP-hard and thus difficult to solve in general.

To tackle this new problem, we first consider the single user setup where there is only  one active user served by a nearby IRS. By exploiting the structure of the simplified  problem, we show that it can be transformed into an integer linear program (ILP), for which the globally optimal solution can be obtained by applying  the branch-and-bound method. To reduce the computational complexity for the optimal solution, we further propose a low-complexity successive refinement algorithm  where the optimal discrete phase shifts of different elements at the IRS are determined one by one in an iterative manner with those of the others being fixed. This algorithm is shown to achieve close-to-optimal performance.  Moreover,  we analytically show that when the number of reflecting elements at the IRS, $N$, increases, the power loss due to discrete phase shifts as compared to the ideal case with  continuous phase shifts approaches  a constant in dB   that depends only on the number of phase-shift levels at each element, but regardless of $N$ as $N\rightarrow \infty$.
As a result, the asymptotic squared power gain of $\mathcal{O}(N^2)$ by  the IRS shown in \cite{wu2018IRS} with continuous phase shifts still holds with discrete phase shifts. Next, we extend the successive refinement algorithm to the general case with multiple users at arbitrary locations in the network,   by considering  the suboptimal  zero-forcing (ZF) based linear precoding at the AP for low-complexity implementation.   Numerical  results are shown to validate our theoretical analysis  and demonstrate the effectiveness of using IRS with practical discrete phase shifts to improve the performance of wireless networks as compared to the case without using IRS.  Furthermore, it is shown that the proposed algorithms for joint AP precoding and IRS discrete phase shifts optimization outperform both the quantization-based and codebook-based schemes with IRS.  Finally, we show that employing practical IRS with even finite-level low-cost phase shifters in wireless networks is able to achieve the same multiuser SINR performance as compared to the conventional large (massive)  MIMO system without using the IRS, but instead using  more active antennas at the AP, thus significantly reducing the system energy consumption as well as hardware cost.

The rest of this paper is organized as follows. Section II introduces the system model and the problem formulation for the IRS-aided  wireless system with discrete phase shifts.
In Sections III and IV, we propose both optimal and suboptimal algorithms to solve the optimization problems in single-user and multiuser cases, respectively.
Section V presents numerical results to evaluate  the performance of the proposed designs. Finally, this  paper  is concluded in Section VI.

\emph{Notations:} Scalars are denoted by italic letters, vectors and matrices are denoted by bold-face lower-case and upper-case letters, respectively. $\mathbb{C}^{x\times y}$ denotes the space of $x\times y$ complex-valued matrices. For a complex-valued vector $\bm{x}$, $\|\bm{x}\|$ denotes its Euclidean norm, $\arg(\bm{x})$ denotes a vector with each entry being the phase of the corresponding entry in $\bm{x}$, and $\text{diag}(\bm{x})$ denotes a diagonal matrix with each diagonal entry being the  corresponding entry in $\bm{x}$.
The distribution of a circularly symmetric complex Gaussian (CSCG) random vector with mean vector  $\bm{x}$ and covariance matrix ${\bm \Sigma}$ is denoted by  $\mathcal{CN}(\bm{x},{\bm \Sigma})$; and $\sim$ stands for ``distributed as''. For a square matrix $\Ss$, ${\rm{tr}}(\Ss)$ and $\Ss^{-1}$ denote its trace and inverse, respectively, while $\Ss\succeq \bm{0}$ means that $\Ss$ is positive semi-definite.  For any general matrix $\A$, $\A^H$,  ${\rm{rank}}(\A)$, and $\A(i,j)$ denote its conjugate transpose, rank, and $(i,j)$th entry, respectively. $\I_M$  denotes an identity matrix  with size $M\times M$. $\mathbb{E}(\cdot)$ denotes the statistical expectation. $ \mathrm{Re}\{\cdot\}$ denotes the real part of a complex number. For a set $\mathcal{K}$, $|\mathcal{K}|$ denotes its cardinality.

\section{System Model and Problem Formulation}
\subsection{System Model}

As shown in Fig. \ref{system:model}, we consider a multiuser multiple-input single-output (MISO) wireless system where an IRS composed of $N$ reflecting elements is deployed to assist in the downlink communication from an AP with $M$ antennas to $K$ single-antenna users. The sets of reflecting elements and users  are denoted by  $\mathcal{N}$ and $\mathcal{K}$, respectively, where $|\mathcal{N}| = N$ and $|\mathcal{K}| = K$.  While this paper focuses on the downlink communication, the results are extendable to the uplink communication as well by exploiting the uplink-downlink  channel reciprocity.  In practice, each IRS is usually attached  with a smart controller that  communicates with the AP via a separate wireless link for coordinating transmission  and exchanging information on e.g.  channel knowledge,  and controls  the phase shifts of all reflecting elements in real time  \cite{JR:wu2019IRSmaga}.
Due to the substantial path loss, we consider  only the signal reflection by the IRS for the first time and ignore the signals that are reflected by it   two or more times.  To characterize the optimal  performance of the IRS-aided wireless system with discrete phase shifts, it is assumed that  the channel state information (CSI) of all channels involved  is perfectly known  at the AP in each channel coherence time, based on the various channel acquisition methods discussed in \cite{JR:wu2019IRSmaga}.

According to \cite{JR:wu2019IRSmaga} and under the assumption of an ideal signal refection model by ignoring the hardware imperfections such as non-linearity and noise,   the  reflected signal by the $n$th element of the IRS, denoted by $\hat{y}_n$, can be expressed as  the multiplication of the corresponding incident  signal, denoted by $\hat{x}_n$,  and a complex reflection coefficient\footnote{For convenience, we use the equivalent baseband signal model to represent the actual signal reflection at IRS  in the RF band.}, i.e.,
\begin{align}
\hat{y}_n = \beta_n e^{j\theta_n}\hat{x}_n,  n \in \mathcal{N},
\end{align}
where  $\beta_n \in [0, 1]$ and $\theta_n\in [0, 2\pi)$ are the reflection amplitude and phase shift of element $n$, respectively.
As such, the IRS with $N$ elements performs a linear mapping from the incident signal vector to a reflected signal vector based on an equivalent  $N\times N$  diagonal phase-shift  matrix $\ttheta$, i.e.,
${\bm{\hat y}} = \ttheta \bm{\hat x}$,  where $\ttheta = \text{diag} (\beta_1 e^{j\theta_1}, \cdots, \beta_N e^{j\theta_N})$,  $\bm{\hat x} = [{\hat x}_1, \cdots, {\hat x}_N]^T$, and $\bm{\hat y} = [{\hat y}_1, \cdots, {\hat y}_N]^T$.  Theoretically, the reflection amplitude of each element  can be adjusted  for different purposes such as channel estimation, energy harvesting, and performance optimization \cite{JR:wu2019IRSmaga}. However, in practice,  it is costly  to implement independent control of the reflection amplitude and phase shift simultaneously; thus, each element is usually  designed to maximize the  signal reflection for simplicity  \cite{JR:wu2019IRSmaga,JR:wu2018IRS,nayeri2018reflectarray,cui2014coding,kaina2014shaping}. As such, we assume $\beta_n=1$, $\forall n\in \mathcal{N}$, in the sequel of this paper.  For ease of practical implementation, we also consider that the phase shift at each element of the IRS can take  only a finite number of discrete values.  Let $b$ denote the number of bits used to indicate the number of phase shift levels $L$ where  $L=2^b$.
For simplicity, we assume that such discrete phase-shift values are obtained by  uniformly quantizing the interval  $ [0, 2\pi)$. Thus, the set of discrete phase-shift values at each element is given by
\begin{align}
\mathcal{F}= \{0,\Delta\theta, \cdots, ( {L}-1)\Delta\theta \},
\end{align}
where $\Delta\theta= 2\pi/L$.

Denote by $\bm{h}^H_{d,k}\in \mathbb{C}^{1\times M}$, $\bm{h}^H_{r,k}\in \mathbb{C}^{1\times N}$, and $\bm{G}\in \mathbb{C}^{N\times M}$ the baseband equivalent channels from the AP to user $k$, the  IRS to user $k$, and the AP to IRS, respectively.  At the AP,  we consider the conventional continuous linear precoding  with  $\bm{w}_k\in \mathbb{C}^{M\times 1}$  denoting the  transmit  precoding vector for user $k$. The complex baseband transmitted signal at the AP can be then expressed as $\bm{x}= \sum_{k=1}^K\bm{w}_ks_k$ where  $s_k$'s denote the information-bearing symbols of users which are modelled as independent and identically distributed (i.i.d.) random variables with zero mean and unit variance. Accordingly, the total transmit power consumed at the AP is given by
\begin{align}
P =\sum_{k=1}^{K}\|\bm{w}_k\|^2.
\end{align}
For user $k$, the signal directly coming  from the AP and that reflected by the IRS are combined at the receiver and thus the received signal can be expressed as
\begin{align}
y_k= ( \bm{h}^H_{r,k}\ttheta \bm{G} +  \bm{h}^H_{d,k}) \sum_{j=1}^K\bm{w}_js_j + z_k,  k \in \mathcal{K},
\end{align}
where $z_k$ denotes i.i.d. additive white Gaussian noise (AWGN) at  user $k$'s receiver  with zero mean and variance $\sigma_k^2$. The SINR of user $k$ is thus given by
\begin{align}\label{eq:SINR}
\text{SINR}_k = \frac{|( \bm{h}^H_{r,k}\ttheta \bm{G}+\bm{h}^H_{d,k})\bm{w}_k |^2}{\sum_{j\neq k}^{K}|( \bm{h}^H_{r,k}\ttheta \bm{G}+\bm{h}^H_{d,k})\bm{w}_j |^2 +  \sigma^2_k},  k \in \mathcal{K}.
\end{align}

{\begin{remark}
\emph{Note that \eqref{eq:SINR} provides a general expression of the SINR for a user at arbitrary location in the IRS-aided single-cell system considered in this paper. While in practice, if user $m$, $m\in \mathcal{K}$, is sufficiently  far from the passive IRS, the reflection of the IRS can be ignored and  its  SINR is approximated by $\text{SINR}_m \thickapprox \frac{|\bm{h}^H_{d,m}\bm{w}_m |^2}{\sum_{j\neq m}^{K}|\bm{h}^H_{d,m}\bm{w}_j |^2 +  \sigma^2_m}$,  which corresponds to the traditional  case where   user $m$ is  served by the AP via transmit precoding only without the IRS. Although the phase-shift matrix $\ttheta$ does not directly affect  the SINR of user $m$ in this case, it can have an indirect effect on it via balancing  the SINRs of those  users nearby the IRS as shown in \eqref{eq:SINR} and thereby adjusting their AP transmit precoding vectors, which then contribute to the interference at user $m$'s receiver.}
\end{remark} }

\subsection{Problem Formulation}
Denote by $\gamma_k>0$ the minimum SINR requirement of user $k,  k\in \mathcal{K}$.  Let $\bm{\theta}= [\theta_1, \cdots, \theta_N]$ and $\W = [{\bm w}_1, \cdots,{\bm w}_K]\in \mathbb{C}^{M\times K}$.   In this paper, we aim to minimize  the total transmit power at the AP by jointly optimizing the  transmit precoders $\W$ at the AP and phase shifts $\bm{\theta}$ at the IRS, subject to the user SINR constraints as well as the IRS discrete phase-shift constraints.   
The corresponding optimization problem is formulated as
\begin{align}
\text{(P1)}: ~~\min_{\W, \bm{\theta}} ~~~&\sum_{k=1}^{K}\|\bm{w}_k\|^2 \\
\mathrm{s.t.}~~~~&\frac{|( \bm{h}^H_{r,k}\ttheta \bm{G}+\bm{h}^H_{d,k})\bm{w}_k |^2}{\sum_{j\neq k}^{K}|( \bm{h}^H_{r,k}\ttheta \bm{G}+\bm{h}^H_{d,k})\bm{w}_j |^2 +  \sigma^2_k}\geq \gamma_k, \forall k \in \mathcal{K}, \label{SINR:constraints}\\
&\theta_n \in \mathcal{F} =\{0,\Delta\theta, \cdots, (L-1)\Delta\theta \}, \forall n \in \mathcal{N}. \label{phase:constraints}
\end{align}
Note that the constraints in \eqref{SINR:constraints} are non-convex due to the coupling of $\W$ and $\bm{\theta}$ in users' SINR expressions. In addition,  the constraints in \eqref{phase:constraints} restrict $\theta_n$'s  to be discrete values. As a result,  problem (P1) is an MINLP which is generally NP-hard,  and there is no standard method for obtaining its globally optimal solution efficiently \cite{so2007approximating,burer2012non}.  One commonly used approach is to first solve problem (P1) with all discrete optimization variables $\theta_n$'s relaxed to their continuous counterparts and then directly quantize each of the obtained continuous phase shifts to its nearest discrete value in $\mathcal{F}$ \cite{wu2018IRS_discrete}. However, even after such relaxation, (P1) is still a non-convex optimization problem  \cite{JR:wu2018IRS}.
Furthermore, such a direct  quantization method  may be ineffective for the practical IRS with low-resolution phase shifters (e.g., $b=1$), especially in the multiuser case with severe co-channel interference (as will be shown later in Section V). Nevertheless, this suboptimal quantization approach will be used in Section III-C to characterize the performance of the IRS with discrete phase shifts  in the regime of asymptotically large $N$ as compared to the ideal case with continuous phase shifts.

\section{Single-User System}\label{sec:III}
First, we consider the single-user setup, i.e., $K=1$, where there is only one user in the considered time-frequency dimension. This corresponds to the practical scenario when orthogonal multiple access (such as time division multiple access) is employed to separate the communications for different users.  Due to the absence of multiuser interference, (P1) is simplified to  (by dropping the user index)
\begin{align}
\min_{\bm{w}, \bm{\theta}} ~~~& \|\bm{w}\|^2  \label{eq:obj}\\
\mathrm{s.t.}~~~~&| (\bm{h}^H_r\ttheta \bm{G}+\bm{h}^H_d )\bm{w}|^2   \geq \gamma \sigma^2,   \label{SINR:constraints:SU} \\
&\theta_n \in \mathcal{F}, \forall n \in \mathcal{N}. \label{phase:constraints:SU0}
\end{align}
For any given phase shifts $\bm{\theta}$, it is known that the maximum-ratio transmission (MRT) is the optimal transmit precoder to problem \eqref{eq:obj} \cite{tse2005fundamentals}, i.e.,
$\bm w^* = \sqrt{P} \frac{(\bm{h}^H_r\ttheta \bm{G}+\bm{h}^H_d  )^H}{\|\bm{h}^H_r\ttheta \bm{G} +\bm{h}^H_d \|}$. By substituting $\bm w^*$ into problem \eqref{eq:obj},
we obtain the optimal transmit power as $P^* =  \frac{\gamma\sigma^2}{\|\bm{h}^H_r\ttheta \bm{G}+ \bm{h}^H_d\|^2}$. As such, minimizing the AP transmit power is equivalent to maximizing the channel power gain of the combined user channel, i.e.,
\begin{align}\label{secIII:p3}
\max_{\bm{\theta}} ~~~&\|\bm{h}^H_r\ttheta \bm{G}+ \bm{h}^H_d\|^2\\
\mathrm{s.t.}~~~~&\theta_n \in \mathcal{F},\forall n \in \mathcal{N}. \label{SecIII:phaseconstraint}
\end{align}
 By applying the change of variables $\bm{h}^H_r\ttheta \bm{G} =\bm{v}^H\bm{\Phi} $ where  $\bm{v} = [e^{j\theta_1}, \cdots, e^{j\theta_N}]^H$ and $\bm{\Phi}=\text{diag}(\bm{h}^H_r)\bm{G} \in \mathbb{C}^{N \times M}$, we have $\|\bm{h}^H_r\ttheta \bm{G} + \bm{h}^H_d\|^2 =\|\bm{v}^H\bm{\Phi}+ \bm{h}^H_d\|^2 $.  Let $\A=\bm{\Phi}\bm{\Phi}^H$ and $\bm{\hat h}_d= \bm{\Phi}\bm{h}_d$. 
   Problem \eqref{secIII:p3} is thus equivalent to
\begin{align}\label{secIII:p4}
\text{(P2)}: ~~\max_{\bm{\theta}} ~~~&\bm{v}^H\A\bm{v} +  2\mathrm{Re}\{ \bm{v}^H \bm{\hat h}_d \}  + \|\bm{h}^H_d\|^2\\
\mathrm{s.t.}~~~~&\theta_n \in \mathcal{F}, \forall n \in \mathcal{N}. \label{SecIII:phaseconstraint}
\end{align}
Although (P2) is still non-convex, we  obtain its optimal and high-quality suboptimal solutions by exploiting its special   structure.

 \subsection{Optimal Solution}
 We first show that problem (P2) can be reformulated as an ILP for which the globally optimal solution can be obtained by applying the  branch-and-bound method. Specifically, we exploit the special ordered set of type 1 (SOS1) \cite{burer2012non} and transform  (P2) into a linear program  with only binary optimization variables. The formal definition of SOS1 is given as follows.
 \begin{definition}
 \emph{A special ordered set of type 1 (SOS1) is a set of vectors with length $N$, each of which  has only one entry  being 1 with the others being  0 \cite{burer2012non}, i.e.,
\begin{align}
 \sum_{i=1}^{N}\x(i) =1,    \x(i) \in\{0, 1\}.
 \end{align}}
 \end{definition}The SOS1 vector is useful in expressing an optimization variable that belongs to a set with discrete variables.
By applying some algebra operations using the fact that $\A$ is positive semi-definite,  the objective function of (P2)  can be expanded as (by ignoring the constant terms)
\begin{align}
\mathcal{G}(\bm{\theta}, \phi_{i,n}) \triangleq &\sum_{i=1}^{N-1}\sum_{n=i+1}^{N}2|\A(i,n)|\left( \cos(\arg(\A(i,n)))\cos(\phi_{i,n} ) - \sin(\arg(\A(i,n)))\sin(\phi_{i,n} ) \right) \nonumber \\
+2\sum_{n=1}^{N}&|\bm{\hat h}_d(n)|( \cos(\arg(\bm{\hat h}_d(n)) )\cos(\theta_n) - \sin( \arg(\bm{\hat h}_d(n))  )\sin(\theta_n)  )+ \sum_{n=1}^{N}\A(n,n),
\end{align}
where $\phi_{i,n} = \theta_i-\theta_n, i, n \in \mathcal{N}$.  As such, problem (P2) is equivalent to
 \begin{align}\label{prob:MILP}
 \max_{\bm{\theta}, \{\phi_{i,n}\}} ~~~&\mathcal{G}(\bm{\theta}, \phi_{i,n})  \\
\mathrm{s.t.}~~~~&\theta_n \in \mathcal{F}, \forall n \in \mathcal{N}, \label{phase:constraints:SU} \\
~~~~& \phi_{i,n} = \theta_i-\theta_n, \forall i, n \in \mathcal{N}. \label{P4:C9}
 \end{align}
 Let ${\mathbf a}=[0,\Delta\theta, \cdots, ( L-1)\Delta\theta]^T$,   $\cc =[1, \cos(\Delta\theta), \cdots, \cos( ( L-1)\Delta\theta)]^T$, and $\s =[1, \sin(\Delta\theta),\\ \cdots, \sin( ( L-1)\Delta\theta)]^T$.
Note that $\cos(\phi_{i,n})$, $\sin(\phi_{i,n})$, $\cos(\theta_n)$, and $\sin(\theta_n)$ in $\mathcal{G}(\bm{\theta}, \phi_{i,n})$ are all non-linear functions with respect to the associated optimization variables.  To overcome this issue, we  introduce an SOS1 binary vector  $\x_n$ for element $n$ of the IRS. Accordingly,  $\theta_n$, $\cos(\theta_n)$, and $\sin(\theta_n)$ can be expressed in the linear form of  $\x_n$, i.e.,
\begin{align}\label{cos:phase}
\theta_n={\mathbf a}^T\x_n,  ~~~\cos(\theta_n)= \cc^T\x_n, ~~~ \sin(\theta_n)= \s^T\x_n.
\end{align}
 Similarly, for the phase-shift difference of two elements $i$ and $n$, i.e., $\phi_{i,n}$, we introduce an SOS1 binary vector  $\y_{i,n}$.  Since the value of $\theta_n$ is chosen from $ \mathcal{ F}$, all the possible values of $\phi_{i,n}= \theta_i-\theta_n\in (-2\pi, 2\pi)$ belongs to the set $\mathcal{\hat F}= \{-(L-1)\Delta\theta ,\cdots, -\Delta\theta, 0,\Delta\theta, \cdots, ( L-1)\Delta\theta \}$.  To overcome the phase ambiguity of $\phi_{i,n}$ with respect to $2\pi$, we further introduce a binary variable  $\varepsilon_{i,n}$ which takes the value of 0 when $ \theta_i-\theta_n$ belongs to $[0, 2\pi)$ and 1 otherwise. As such,   all the possible values of  $\phi_{i,n}$ are restricted to the set $\mathcal{ F}$ and accordingly we  have
\begin{align}\label{cos:phase:diff}
\phi_{i,n}= {\mathbf a}^T\y_{i,n}-  2\pi\varepsilon_{i,n},  ~~~\cos(\phi_{i,n})= \cc^T\y_{i,n}, ~~~ \sin(\phi_{i,n})= \s^T\y_{i,n}.
\end{align}
Substituting \eqref{cos:phase} and  \eqref{cos:phase:diff} into problem \eqref{prob:MILP}, we obtain the following optimization problem
 \begin{align}\label{prob:MILP2}
 \max_{\{\x_n\},\{\y_{i,n}\}, \{\varepsilon_{i,n}\} } ~~~&\sum_{i=1}^{N-1}\sum_{n=i+1}^{N}2|\A(i,n)|\left( \cos(\arg(\A(i,n))) \cc^T - \sin(\arg(\A(i,n))) \s^T \right)\y_{i,n} \nonumber \\
&+2\sum_{n=1}^{N}|\bm{\hat h}_d(n)|( ( \cos(\arg(\bm{\hat h}_d(n)) ) \cc^T - \sin( \arg(\bm{\hat h}_d(n))  )\s^T)\x_n  \\
\mathrm{s.t.}~~~~& {\mathbf a}^T(\x_i -\x_n) +2\pi \varepsilon_{i,n} ={\mathbf a}^T\y_{i,n}, \forall i, n \in \mathcal{N},  \label{P4:C9} \\
~~~~& \x_n \in {SOS1},    \y_{i,n} \in {SOS1}, \varepsilon_{i,n}\in \{0, 1\}, \forall i, n  \in \mathcal{N}.
\end{align}
It is not difficult to show  that problem \eqref{prob:MILP2} is an ILP and thus can be optimally solved  by applying the branch-and-bound method \cite{burer2012non}.
\subsection{Suboptimal Solution}\label{suboptimal:solution:SU}
Although the optimal solution to (P2)  can be obtained as in the previous subsection,  the worst-case complexity  is still exponential over $N$  due to its fundamental  NP-hardness. To reduce the computational complexity,  we propose in this subsection an efficient successive refinement algorithm to solve (P2) sub-optimally, which will also be extended to the general  multiuser case in Section \ref{multiuser:sec}.  Specifically, we alternately optimize each of the $N$ phase shifts in an iterative manner by fixing the other $N-1$ phase shifts, until  the convergence is achieved.

The key to solving (P2) by applying the successive refinement algorithm  lies in the observation that for a given $n\in \mathcal{N}$,  by fixing $\theta_\ell$'s, $\forall \ell\neq n, \ell\in \mathcal{N}$, the objective function of (P2) is  linear with respect to $e^{j\theta_n}$, which can be written as
\begin{align}\label{obj:phase}
&2\mathrm{Re}\left\{e^{j\theta_n}\zeta_n\right\} + \sum_{\ell\neq n}^{N}\sum_{i\neq n}^{N}\A(\ell,i)e^{j(\theta_\ell-\theta_i)} +C,
\end{align}
where $\zeta_n$ and $C$ are constants  given by
\begin{align}
\zeta_n &= \sum_{\ell\neq n}^{N}\A(n,\ell)e^{-j\theta_\ell} +  \bm{\hat h}_{d}(n) \triangleq  |\zeta_n |e^{-j\varphi_n},  \label{obj:phase2} \\
C  &=  \A(n,n)+ 2\mathrm{Re} \Big\{\sum_{\ell\neq n}^{N}e^{j\theta_\ell}\bm{\hat h}_{d}(\ell)\Big\}   + \|\bm{\hat h}_d\|^2.
\end{align}
 Based on \eqref{obj:phase} and \eqref{obj:phase2}, it is not difficult to verify that the optimal $n$th phase shift is given by
\begin{align}\label{optimal:phase:SU}
\theta^*_n = \arg \min_{\theta \in \mathcal{F}  } | \theta -\varphi_n|.
\end{align}
By successively  setting the phase shifts of all elements based on \eqref{optimal:phase:SU} in the order from $n=1$ to $n=N$ and then repeatedly, the objective value of (P2) will be non-decreasing over the iterations. On the other hand,  the optimal objective value of  (P2) is upper-bounded by a finite value, i.e.,
\begin{align}
\bm{v}^H\A\bm{v} + 2\mathrm{Re}\{\bm{v}^H \bm{\hat h}_d\}  + \|\bm{h}^H_d\|^2\overset{(a)}{\leq} N\lambda_{\max}(\A) + 2\sum_{n=1}^{N}|\bm{\hat h}_d^H(n)| + \|\bm{h}^H_d\|^2,
\end{align}
where $\lambda_{\max}(\A)$ is the maximum eigenvalue of $\A$ and the inequality $(a)$ holds due to $\bm{v}^H\A\bm{v}\leq N\lambda_{\max}(\A)$ and $\mathrm{Re}\{\bm{v}^H \bm{\hat h}_d\}\leq \sum_{n=1}^{N}|\bm{\hat h}_d^H(n)|$.
Therefore,  the proposed algorithm is guaranteed to converge to a locally optimal solution. With the converged discrete phase shifts, the minimum transmit power $P^*$ can be obtained accordingly.

\subsection{Performance Analysis for IRS with Asymptotically Large $N$}
Next, we characterize the scaling law of the average received power at the user with respect to the number of reflecting elements, $N$, as $N\rightarrow \infty$ in an IRS-aided system with discrete phase shifts. For simplicity, we assume $M=1$ with $\G\equiv \bm{g}$ to obtain essential insight. Besides, the signal received at the user from the AP-user  link can be practically ignored for asymptotically large $N$ since in this case, the reflected signal power dominates in the total received power. Thus, the user's average received  power with $b$-bit phase shifters is approximately given by $P_r(b)\triangleq  \mathbb{E}(|h^H|^2) = \mathbb{E}(|\bm{h}^H_r\ttheta \bm{g}|^2)$ where $\bm{\theta}$ is assumed to be obtained  by quantizing each of the optimal continuous phase shifts obtained in \cite{wu2018IRS} to its nearest discrete value in $\mathcal{F}$.
\begin{proposition}\label{scaling:law}
\emph{\!\! Assume $\bm{h}^H_{r} \!\sim \!\mathcal{CN}(\bm{0},\!\varrho^2_h{\bm I})$ and \! $\bm{g} \!\sim\! \mathcal{CN}(\bm{0},\!\varrho^2_g{\bm I})$.   As $N\rightarrow \infty$,  we have}
\begin{align}\label{ratio}
\eta (b) \triangleq  \frac{P_r(b)}{P_r(\infty)} = \Big(\frac{2^b}{\pi}\sin\left(\frac{\pi}{2^b}\right)\Big)^2.
\end{align}
\end{proposition}
\begin{proof}
The combined user channel can be expressed as
\begin{align}
{h}^H=   \bm{h}^H_{r} \ttheta\bm{g} =\sum_{n=1}^N |\hh^H_{r}(n)||\g(n)|e^{j(\theta_n+\phi_n+\psi_n)},
\end{align}
where $\hh^H_{r}(n)=|{\hh}^H_{r}(n)|e^{j\phi_n}$ and $\g(n)=|\g(n)|e^{j\psi_n}$, respectively. Since $|{\hh}^H_{r}(n)|$ and $|\g(n)|$ are statistically independent and follow Rayleigh distribution with mean values $\sqrt{\pi}\varrho_h/2$ and  $\sqrt{\pi}\varrho_g/2$, respectively, we have $\mathbb{E}(|{\hh}^H_{r}(n)||\g(n)|)=\pi\varrho_h\varrho_g/4$.  Since $\phi_n$ and $\psi_n$ are randomly and uniformly distributed in $[0, 2\pi)$, it follows that $\phi_n+\psi_n$ is uniformly distributed in $[0, 2\pi)$ due to the periodicity over $2\pi$.  As such, the optimal continuous phase shift is given by $\theta^{\star}_n = -(\phi_n +\psi_n)$, $n\in \mathcal{N}$ \cite{wu2018IRS}, with the corresponding quantized discrete phase shift denoted by  $\hat{\theta}_n$ which can be obtained similarly as \eqref{optimal:phase:SU}. 
  Define $\bar \theta_n=\hat{\theta}_n- \theta^{\star}_n=\hat{\theta}_n + \phi_n+\psi_n$ as the quantization error. As $\hat{\theta}_n$'s in  $\mathcal{F}$ are equally spaced,   it follows that $\bar \theta_n$'s are independently and uniformly distributed in $[-\pi/2^b, \pi/2^b)$. Thus, we have
{\begin{align}
\vspace{-2mm}
&\!\!\!\mathbb{E}(|{h}^H|^2)\!  =\!  \mathbb{E}\left(\left|\sum_{n=1}^N |{\hh}^H_{r}(n)||\g(n)|e^{j{\bar \theta_n}}  \right|^2 \right ) \! \!=\!  \mathbb{E} \!\left(   \sum_{n=1}^N |{\hh}^H_{r}(n)|^2|\g(n)|^2 \right. \nonumber \\
~&~~~~~~~~~~~~\quad~\left.   +        \sum_{n=1}^N \sum_{i\neq n}^N |{\hh}^H_{r}(n)||\g(n)||{\hh}^H_{r}(i)||\g(i)|e^{j{\bar \theta_n}-j{\bar \theta_i}}   \right).
\end{align}}Note that ${\hh}^H_{r}(n)$, $\g(n)$, and $e^{j\bar \theta_n}$ are independent with each other, with $\mathbb{E}\left(   \sum_{n=1}^N |{\hh}^H_{r}(n)|^2|\g(n)|^2\right)= N\varrho^2_h\varrho^2_g$ and $\mathbb{E}(e^{j\bar \theta_n})=\mathbb{E}(e^{-j\bar \theta_n})= 2^b/\pi\sin\left(\pi/2^b\right)$. It then follows that
\begin{align}\label{eq:pow}
\!\!\!\!P_r(b)&\! =\!N\varrho^2_h\varrho^2_g + N(N-1)\frac{\pi^2\varrho^2_h\varrho^2_g}{16}\Big(\frac{2^b}{\pi}\sin\left(\frac{\pi}{2^b}\right)\Big)^2.
\end{align}
For $b\geq 1$, it is not difficult to verify that $2^b/\pi\sin\left(\pi/2^b\right)$ increases with $b$  monotonically and satisfies
\begin{align}
\lim_{b \rightarrow \infty}\frac{2^b}{\pi}\sin\left(\frac{\pi}{2^b}\right)   = 1,
\end{align}
where $b \rightarrow \infty$ corresponds to the case of continuous phase shifts without quantization.  As a result, the ratio between $P_r(b)$ and $P_r(\infty)$ is given by \eqref{ratio} as $N\rightarrow \infty$, which completes the proof.
\end{proof}
Proposition \ref{scaling:law} provides a quantitative measure of  the user received power loss with discrete phase shifts as compared to the ideal case with continuous phase shifts. It is observed that  as $N\rightarrow \infty$, the power ratio $\eta (b)$ depends only on the number of discrete phase-shift levels, $2^b$,  but is regardless of $N$. This result implies that using a practical IRS  even with discrete phase shifts, the same asymptotic squared power gain of $\mathcal{O}(N^2)$ as that with  continuous phase shifts shown in \cite{wu2018IRS} can still be achieved (see  \eqref{eq:pow} with $N\rightarrow \infty$).  As such,  the design of  IRS hardware and control module can be greatly simplified by using discrete phase shifters, without compromising the performance in the large-$N$ regime.  Since $\eta (1)  =-3.9$ dB, $\eta (2)= -0.9$ dB,  and $\eta (3)=-0.2$ dB as shown in Table \ref{table2},  using 2 or 3-bit phase shifters is practically sufficient to achieve close-to-optimal  performance with only approximately 0.9 dB or 0.2 dB power loss. In general,  there exists an interesting cost tradeoff between the number of reflecting elements $(N)$ and the resolution of phase shifters $(2^b)$ used at the IRS. For example, one can use more reflecting elements (larger $N$) each with a lower-resolution (smaller $b$) phase shifter (thus lower cost per element), or less number of elements (smaller $N$) with higher-resolution (larger $b$) phase-shifters (i.e., higher cost per element), to achieve the same received power at the user.  As such, $N$ and $b$ can be flexibly set in practical systems based on the required performance as well as the manufacturing cost of each reflecting element and that of its phase shifter component so as to minimize the total cost of the IRS.
\begin{table*}[!t]
\caption{The power loss of using IRS with discrete phase shifts.}\label{table2}
\centering
\small
\newcommand{\tabincell}[2]
\small
\begin{tabular}{|m{5cm} |m{1cm}|m{1cm}|m{1cm}|m{3.6cm}|}
  \hline
{Number of control bits:} {$b$}& {$b=1$} &{$b=2$} & {$b=3$} & $b= \infty$~\textcolor{white}{xxxxxx xxxxxx} (continuous phase shifts)     \\ \hline
{Power loss:}                 $1\big/\Big(\frac{2^b}{\pi}\sin\left(\frac{\pi}{2^b}\right)\Big)^2$ &  $3.9$ dB  & $0.9$ dB & $0.2$ dB & $0$ dB    \\ \hline
\end{tabular}\vspace{-0.5cm}
\end{table*}

\section{Multiuser System}\label{multiuser:sec}
In this section, we study the general multiuser setup where multiple  users share the same time-frequency dimension for communications (e.g., in space division multiple access) and they are located at arbitrary locations in the single-cell network among which only some are aided by the nearby IRS in general.  For this general setup, we propose two algorithms to obtain the optimal and suboptimal solutions to (P1), respectively.
\subsection{Optimal Solution}
For any given phase shifts ${\bm \theta}$,  the combined channel from the AP to user $k$ is denoted by ${\bm{h}}^H_k\triangleq  \bm{h}^H_{r,k}\ttheta\bm{G}+\bm{h}^H_{d,k}$.
Thus, problem (P1) is reduced to
\begin{align}
\text{(P3)}: ~~\min_{\W} ~~~&\sum_{k=1}^{K}\|\bm{w}_k\|^2 \\
\mathrm{s.t.}~~~~&\frac{|{\bm{h}}^H_k\bm{w}_k |^2}{\sum_{j\neq k}^{K}|{\bm{h}}^H_k\bm{w}_j |^2 +  \sigma^2_k}\geq \gamma_k, \forall k \in \mathcal{K}.\label{P2:SINR}
\end{align}
Note that (P3) is the conventional power minimization problem in the multiuser MISO downlink broadcast channel, which can be efficiently and optimally solved by using the fixed-point iteration algorithm based on the uplink-downlink duality \cite{wiesel2006linear,schubert2004solution,luo2006introduction}.  Specifically,  the optimal solution is  known as the   minimum mean
squared error (MMSE) based linear precoder  given by
\begin{align}\label{eq:MISObf}
\w^*_k = \sqrt{p_k}\hat{\w}^*_k, \forall k \in \mathcal{K}.
\end{align}
where
\begin{equation}\label{eq:powerallocation}
\begin{bmatrix}
p_1  \\
\vdots\\
p_K  \\
\end{bmatrix}= {\bm Q}^{-1}\begin{bmatrix}
\sigma^2_1  \\
\vdots\\
\sigma^2_K    \\
\end{bmatrix},~~
{\bm Q}(i,j)=
\left\{
\begin{aligned}
&\frac{1}{\gamma_i}|\hh^H_i\hat{\w}^*_i|^2,  && ~~i=j, \\
&-{|\hh^H_i\hat{\w}^*_j|^2}, && ~~i\neq j,  \forall i, j \in \mathcal{K}, \\
\end{aligned}
\right.
\end{equation}
\begin{align}\label{bf:direction}
\hat{\w}^*_k =\frac{ (\I_{M} + \sum_{i=1}^{K} \frac{\lambda_i}{\sigma_i^2}{\bm{h}}_i{\bm{h}}^H_i )^{-1}{\bm{h}}_k}{ ||   (\I_{M} + \sum_{i=1}^{K} \frac{\lambda_i}{\sigma_i^2}{\bm{h}}_i{\bm{h}}^H_i )^{-1}{\bm{h}}_k     || }, \forall k,
\end{align}
\begin{align}\label{lambda}
\lambda_k = \frac{\sigma^2_k}{(1+\frac{1}{\gamma_k} )\hh^H_k(  \I_{M} +\sum_{i=1}^{K} \frac{\lambda_i}{\sigma_i^2} \hh_i\hh_i^H)^{-1}\hh_k}, \forall k.
\end{align}
First, all $\lambda_k$'s can be obtained by using the fixed-point algorithm to solve $K$ equations in \eqref{lambda}. With $\lambda_k$'s, $\hat{\w}^*_k$'s can be obtained from \eqref{bf:direction} and then  $p_k$'s can be obtained from \eqref{eq:powerallocation}. Finally,   ${\w}^*_k$'s are obtained  by using \eqref{eq:MISObf} with $\hat{\w}^*_k$'s and  $p_k$'s.

 As shown in \eqref{eq:MISObf}-\eqref{lambda},  the optimal transmit precoder $\W$ cannot be expressed as a closed-form expression of ${\bm \theta}$ as in the single-user case (i.e., the MRT precoder in Section III) and thus transforming (P1) into an ILP is impossible to our best knowledge. As such, the globally optimal phase shifts to (P1) can only be obtained by the exhaustive search method. Specifically, we can  search all the possible cases of ${\bm \theta}$ and for each case, we solve (P3) to obtain the corresponding transmit power at the AP.  The globally optimal ${\bm \theta}$ is then given by the one that achieves the minimum AP transmit power. As the optimal algorithm requires computing the MMSE precoder $\W$ and exhaustively searching the phase shifts ${\bm \theta}$,   the total complexity for it can be shown to be  $\mathcal{O}(L^N( I_{itr}(KM^2+ M^3)+ K^3+K^2M + KMN ))$ where $I_{itr}$ denotes the number of iterations required for obtaining $\lambda_k$'s in \eqref{lambda} in each case (which is observed  to increase with $K$ linearly in our simulations).

\subsection{Suboptimal Solution}
To reduce the computational complexity of the optimal solution, we extend the successive refinement algorithm in Section  \ref{suboptimal:solution:SU} to the multiuser case, assuming  $M\geq K$.  Specifically, the suboptimal ZF-based linear precoder is employed  at the AP to  eliminate the multiuser interference and meet all the SINR requirements. Then the phase shifts at the IRS are successively refined to minimize the total transmit power at the AP.

With  the combined channel ${\bm{h}}^H_k$'s, $k\in \mathcal{K}$,  the corresponding ZF constraints are given by ${\bm{h}}^H_k\w_j=(\bm{h}^H_{r,k}\ttheta\bm{G}+\bm{h}^H_{d,k}){\bm w}_j =0$, $\forall j\neq k, j\in \mathcal{K}$.  
Let ${\HH}^{H} =  \HH_{r}^{H}\ttheta\G + \HH_{d}^{H}$, where ${\HH}^{H}_r=[{\bm h}_{r,1}, \cdots,{\bm h}_{r,K}]^H$ and ${\HH}_d^{H}=[{\bm h}_{d,1}, \cdots,{\bm h}_{d,K}]^H$.
 With those additional constraints, it is not difficult to verify that the optimal transmit precoder $\W$  to (P1)  is given by the pseudo-inverse of the combined channel $\HH^H$ with proper power allocation among different users, i.e.,
\begin{align}\label{zf:bf}
\W = \HH( \HH^H \HH)^{-1}\PP^{\frac{1}{2}},
\end{align}
 where $\PP = \text{diag}(p_1,\cdots, p_K)$ is the power allocation matrix. By substituting $\W$ into \eqref{SINR:constraints} in (P1), the SINR constraint of user $k$  is reduced to $\frac{p_k}{\sigma^2_k} \geq \gamma_k, \forall k \in \mathcal{K}$. Since this constraint should be  met with equality at the optimal solution to (P1), we have ${p_k}= \sigma^2_k\gamma_k,  k\in \mathcal{K}$.  The total transmit power at the AP is then given by
 \begin{align}
\sum_{k=1}^{K}\|\bm{w}_k\|^2 &= {\rm{tr}}(\W^H\W) = {\rm{tr}}(\PP^\frac{1}{2}( \HH^H \HH)^{-1}\PP^\frac{1}{2} )\nonumber \\
 &\overset{(a)}{=}{\rm{tr}}(\PP ( \HH^H \HH)^{-1})=  {\rm{tr}}(\PP( (\HH_{r}^{H}\ttheta\G + \HH_{d}^{H})(\HH_{r}^{H}\ttheta\G + \HH_{d}^{H})^H )^{-1}  ),
 \end{align}
 where $(a)$ is due to  the  fact that ${\rm{tr}}(\A\B)=  {\rm{tr}}(\B\A)$ for any matrices $\A$ and $\B$ with appropriate dimensions.
As a result, (P1) is transformed to
\begin{align}
\text{(P5)}: ~~\min_{\bm{\theta}} ~~~& {\rm{tr}}\left(\PP( (\HH_{r}^{H}\ttheta\G + \HH_{d}^{H})(\HH_{r}^{H}\ttheta\G + \HH_{d}^{H})^H )^{-1} \right)\triangleq  P({\bm \theta}) \\
\mathrm{s.t.}~~~~&\theta_n \in \mathcal{F},\forall n \in \mathcal{N}. \label{eq:modulus2}
\end{align}
Note that for $K=1$,  (P5) is equivalent to the combined channel power gain maximization problem, i.e.,  (P2),  in the single-user case considered in Section III.
However, in the multiuser case,  (P5) becomes more involved than (P2)  due to the matrix inverse operation that results in a more complicated  non-convex objective function $P({\bm \theta})$.
Nevertheless, by fixing any $N-1$ phase shifts in each iteration,  we can find the optimal solution of the remaining  discrete phase shift via one-dimensional search over $\mathcal{F}$, i.e.,
\begin{align}\label{optimal:phase}
\theta^*_n = \arg \min_{\theta_n \in \mathcal{F}} P({\bm \theta}).
\end{align}
Note that for the above problem, if a rank-deficient  channel matrix, i.e., ${\rm{rank}}(\HH) < K$, is encountered for the optimization of some   $\theta_n$, then the corresponding value of $P({\bm \theta})$ is  set as positive infinity for tractability.
Considering that  the number of discrete phase-shift values in $\mathcal{F}$  is generally limited in practice \cite{cui2014coding,kaina2014shaping} (e.g., $L=2$ for $b=1$ or $L=4$ for $b=2$), the  one-dimensional search in \eqref{optimal:phase} is very efficient.
The above procedure is repeated until the  fractional decrease of $P({\bm \theta})$ is less than a sufficiently small threshold.
It can be similarly  shown as in Section III-B that $P({\bm \theta})$ is  lower-bounded by a finite value and thus the convergence of the proposed ZF-based successive refinement algorithm  is  guaranteed.

In contrast to the optimal solution, the suboptimal solution is based on the ZF precoder at the AP and the successive refinement algorithm for finding the phase shifts at the IRS, thus its  complexity is given by $\mathcal{O}(\hat{I}_{itr}L(K^3+K^2M+KMN))$,  where $\hat{I}_{itr}$ denotes the number of iterations required for achieving convergence of the successive refinement algorithm. Note that
$K^3+K^2M + KMN\leq I_{itr}(KM^2+ M^3)+K^3+K^2M + KMN$ always holds, and $\hat{I}_{itr}L$ is usually much less than $L^N$ in practice based on our simulations.  Thus, the proposed suboptimal algorithm is computationally much more efficient  for IRS with small $L$ and large $N$, as compared to the optimal algorithm.
%
\section{Numerical Results}

 In this section, we provide numerical results to validate our analysis as well as the effectiveness of the proposed algorithms.  A three-dimensional (3D) coordinate is considered as shown in Fig. \ref{simulation:setup:SU},  where a uniform linear array (ULA) at the AP and a uniform rectangular array (URA) at the IRS  are located in $x$-axis and $y$-$z$ plane, respectively. The reference antenna/element at the AP/IRS are respectively  located at $(d_x, 0, 0)$ and  $(0, d_y, 0)$, where in both cases  a half-wavelength spacing is assumed among adjacent antennas/elements.  For the IRS, we set $N=N_{y}N_{z}$ where $N_{y}$ and $N_{z}$ denote the number of reflecting elements along  $y$-axis and  $z$-axis, respectively.   For the purpose of exposition, we fix $N_y=4$ and increase $N_z$ linearly with $N$.  The distance-dependent channel path loss is modeled as
\begin{align}\label{pathloss}
\eta(d) = C_0\left( \frac{d}{d_0} \right)^{-\alpha},
\end{align}
where $C_0$ is the path loss at the reference distance $d_0=1$ meter (m), $d$ denotes the link distance, and $\alpha$ denotes the path loss exponent. Each antenna at the AP is assumed to have an isotropic radiation pattern and thus the antenna gain is 0 dBi. In contrast, as the IRS reflects signals only in its front half-sphere, each reflecting element is assumed to have a 3 dBi gain for fair comparison.
To account for  small-scale fading, we assume the Rician fading channel model for all channels involved.  For example, the AP-IRS channel $\G$ can be expressed as
 \begin{align}
 \G = \sqrt{\frac{ \beta_{\rm AI} }{1+  \beta_{\rm AI} }}\G^{\rm LoS} +  \sqrt{\frac{1}{1+\beta_{\rm AI}}}\G^{\rm NLoS},
 \end{align}
 where $\beta_{\rm AI}$ is the Rician factor, and $\G^{\rm LoS} $ and $\G^{\rm NLoS}$ represent  the deterministic LoS (specular) and Rayleigh fading components, respectively. Note that the above model is simplified  to   Rayleigh fading channel when  $ \beta_{\rm AI} = 0$ or LoS channel when $ \beta_{\rm AI} \rightarrow  \infty$.
  The entries in $ \G$ are then multiplied by the square root of the distance-dependent path loss in \eqref{pathloss} with the path loss exponent denoted by   $\alpha_{\rm AI}$. The AP-user and IRS-user channels are similarly  generated by following the above procedure.  The path loss exponents of the AP-user and IRS-user links are  denoted by $\alpha_{\rm Au}$ and  $\alpha_{\rm Iu}$, respectively, and the corresponding Rician factors are denoted by  $\beta_{\rm Au}$ and $\beta_{\rm Iu}$, respectively.
  In practice, the IRS is usually deployed to serve the users that suffer from severe signal attenuation in the AP-user channel and thus we set $\alpha_{\rm Au}=3.5$ and $\beta_{\rm Au}=0$, while their counterparts for AP-IRS and IRS-user channels will be properly specified later depending on the scenarios.
  Without loss of generality, we assume that all users have the same SINR target, i.e., $\gamma_k=\gamma,  k \in \mathcal{K}$. {The stopping threshold for the successive refinement algorithms is set as $10^{-4}$.}  Other system parameters are set as follows unless specified later:  $C_0= -30$ dB,  $\sigma_k^2=-90$\,dBm, $k \in \mathcal{K}$, $d_x= 2$ m, and  $d_y= 50$ m.

 \begin{figure}[!t]
\centering
\includegraphics[width=0.55\textwidth]{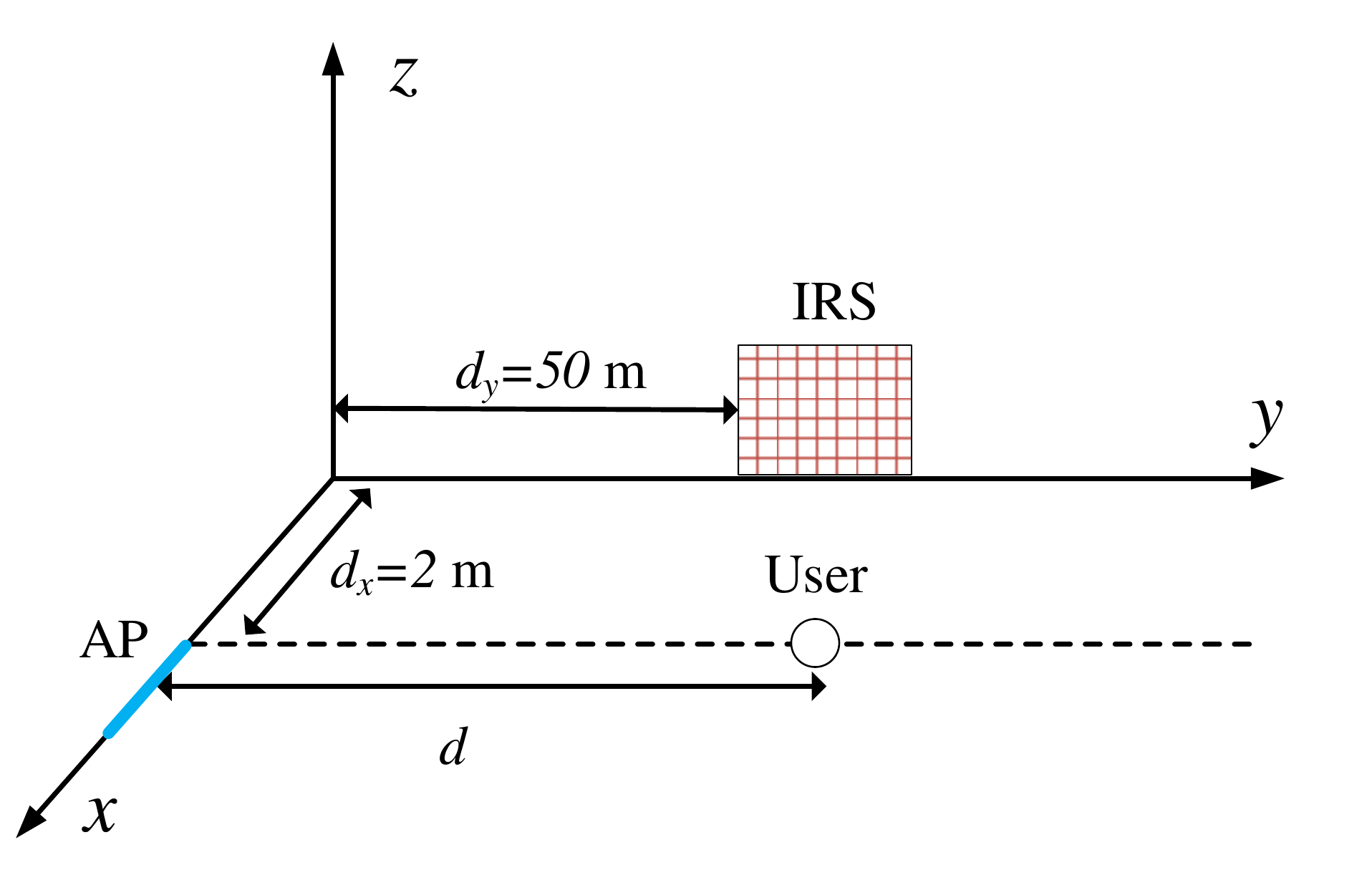}
\caption{Simulation setup of the single-user case. } \label{simulation:setup:SU}\vspace{-5mm}
\end{figure}

\subsection{Single-User System}
\subsubsection{Performance Comparison with  Benchmark Schemes}

We consider that one single  user lies on a  line that is in parallel to $y$-axis shown in Fig. \ref{simulation:setup:SU}, with its location denoted by ($d_x$, $d$, 0).
By varying the value of $d$, the distances of AP-user and IRS-user links change accordingly and we examine the minimum transmit power required for serving the user with a given SNR target. The channel parameters are set as $\alpha_{\rm AI}=2.2$, $\alpha_{\rm Iu}=2.8$, $\beta_{\rm AI}=0$, and $\beta_{\rm Iu}= \infty$.
We compare the following schemes: 1) Lower bound: solve (P1) with $b\rightarrow \infty$ or continuous phase shifts  by using semidefinite relaxation (SDR) with Gaussian randomization which has been shown to achieve near-optimal performance in \cite{wu2018IRS};  2) Optimal algorithm: solve problem \eqref{prob:MILP2} by using the branch-and-bound method; 3) Successive refinement: use the proposed suboptimal algorithm in Section III-B;  4) Quantization scheme: quantize the continuous phase shifts obtained in 1) to their  respective nearest values in $\mathcal{F}$; 5) Codebook based scheme (explained later) which is also used as the phase-shift initialization required in scheme 3);  6)  Benchmark scheme without using the IRS by setting  $\bm{w} = \sqrt{\gamma \sigma^2}{\bm{h}_d}/{\|\bm{h}_d\|^2}$.  For schemes 2)-5), we set $b=1$.  For the  codebook based scheme, we adopt the widely used Hadamard matrix \cite{liu2018joint},  whose entries are either $1$ or $-1$, thus corresponding to the phase shift of $0$ or $\pi$. Besides, its columns are mutually orthogonal and thus span the whole $N$-dimensional space. The codebook based scheme starts by using each of the $N$ columns of the Hadamard matrix as the phase-shift vector ${\bm \theta}$ and then selects the one resulting in the minimum transmit power at the AP. Note that it generally outperforms the scheme with fixed phase shifts at the IRS since the latter can be considered as a special case of the former with only one single vector in the codebook.

 \begin{figure}[!t]
\centering
\includegraphics[width=0.65\textwidth]{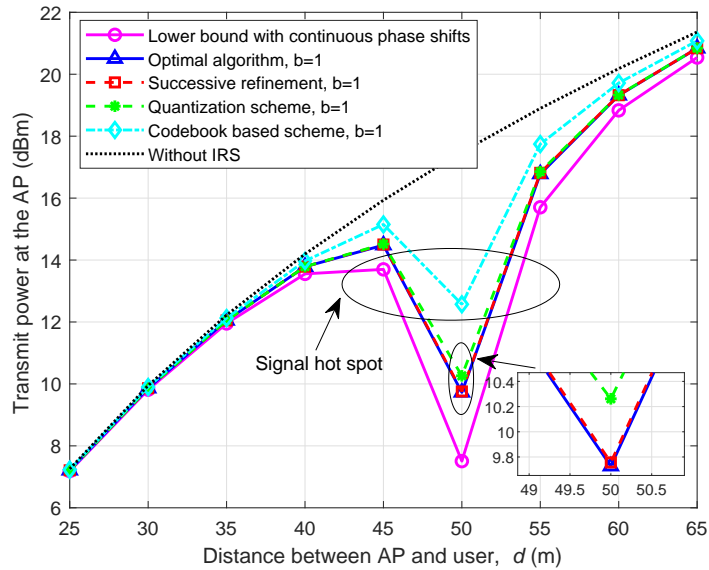}
\vspace{-0.4cm}
\caption{AP transmit power versus AP-user distance. } \label{simulation:distance} \vspace{-4.5mm}
\end{figure}

In Fig. \ref{simulation:distance}, we compare  the  transmit power required at the AP for the above schemes versus the AP-user distance by setting $M=4$,  $N=16$, and $\gamma=25$ dB.
First,  it  is observed that the required transmit power of using 1-bit phase shifters is significantly lower than that without the IRS when the user locates in the vicinity of the IRS. This demonstrates the practical usefulness of IRS in creating a ``signal hot spot'' even with very coarse and low-cost phase shifters.
 Moreover, one can observe that using the IRS with 1-bit phase shifters suffers performance loss compared to the transmit power lower bound with continuous phase shifts. This is expected since due to  discrete phase shifts,  the  multi-path signals  from the AP including those reflected and non-reflected by the IRS  cannot be perfectly aligned in phase at the receiver, thus resulting in a performance loss.
Finally,  it is observed that the proposed successive refinement algorithm and  quantization scheme both achieve near-optimal  performance in this single-user case, and they significantly outperform the codebook based scheme. This  demonstrates the advantage of optimizing phase shifts based on the actual channels over only selecting them from a set of  pre-defined phase shift vectors in a codebook. 

\begin{figure}[!t]
\centering
\includegraphics[width=0.65\textwidth]{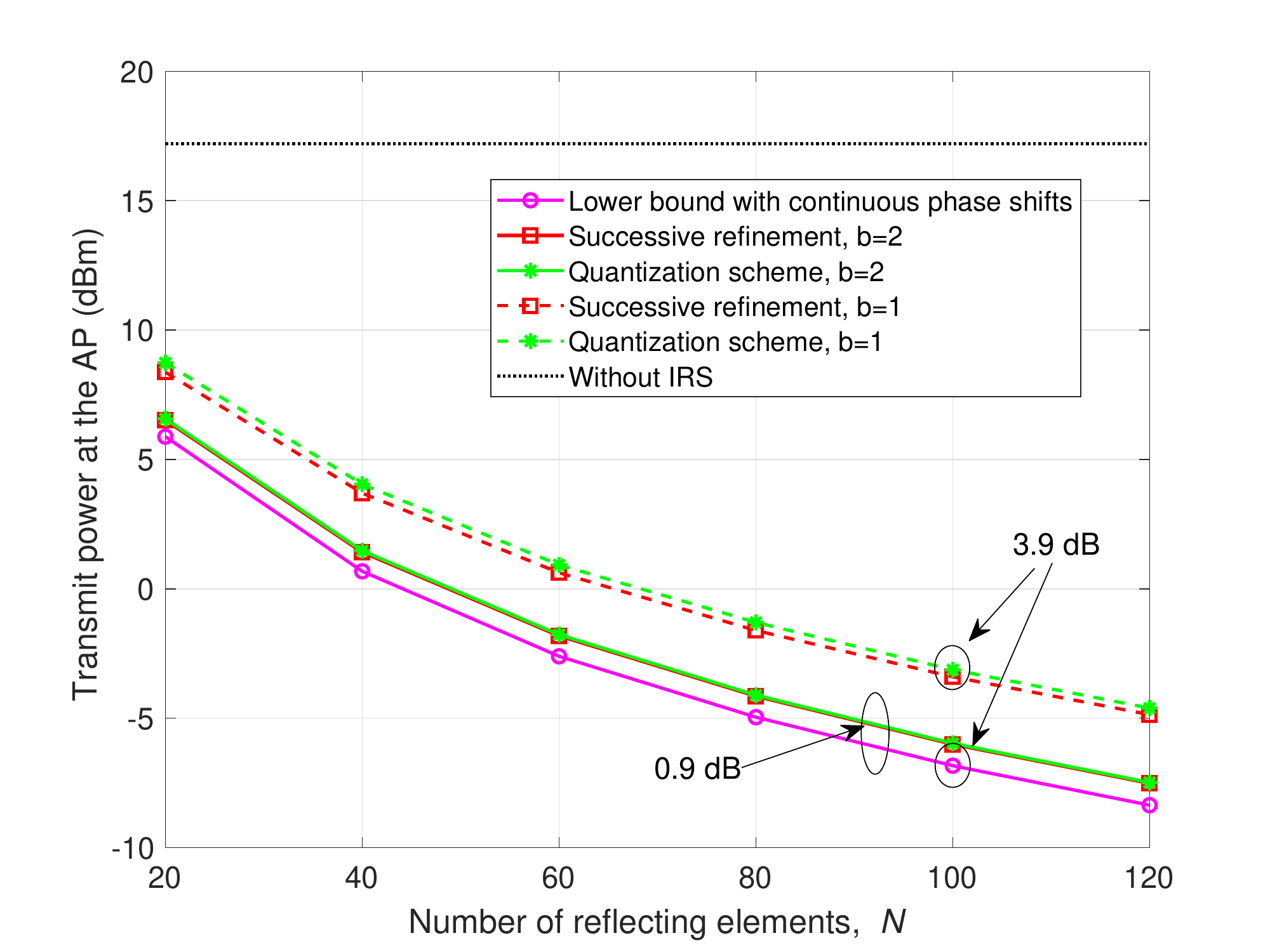}
\vspace{-0.4cm}
\caption{AP transmit power versus the number of reflecting  elements.  } \label{simulation:N} \vspace{-5mm}
\end{figure}

\subsubsection{Impact of Discrete Phase Shifts}

To validate the theoretical analysis  in Proposition \ref{scaling:law}, we plot in Fig.  \ref{simulation:N} the AP transmit power  versus the number of reflecting elements $N$ at the IRS when $d=50$ m. In particular, we consider both $b=1$ and $b=2$ for discrete phase shifts at the IRS. Other parameters are set the same as those in Fig. \ref{simulation:distance}.  From Fig. \!\ref{simulation:N}, it is  observed that as $N$ increases, the performance gap between the quantization scheme (for both $b=1$ and $b=2$) and the lower bound (for $b=\infty$) first increases and then approaches a constant that is determined by $\eta(b)$ given in \eqref{ratio} (i.e., $\eta (1)  =-3.9$ dB  and $\eta (2)=-0.9$ dB shown in Table \ref{table2}). This is expected since when $N$ is moderate,  the signal power of the AP-user link is comparable to that of the IRS-user link, thus the misalignment of multi-path signals due to discrete phase shifts becomes more pronounced with increasing $N$. However, when $N$ is sufficiently large such that the reflected signal power by the IRS  dominates in the total received power at the user, the performance loss arising from the phase quantization error converges to that in accordance with  the asymptotic analysis given in Proposition  \ref{scaling:law}.  In addition, one can observe that in this case the gain achieved by the successive refinement algorithm over the quantization scheme is more evident with $b=1$ compared to $b=2$. It is worth pointing out that the quantization scheme needs to first obtain the continuous phase shifts by invoking the  semidefinite program (SDP)  solver  \cite{wu2018IRS} and thus has a higher complexity than the  successive refinement algorithm.

\begin{figure}[ht]
\centering
~~~~~~~~\includegraphics[width=0.4\textwidth]{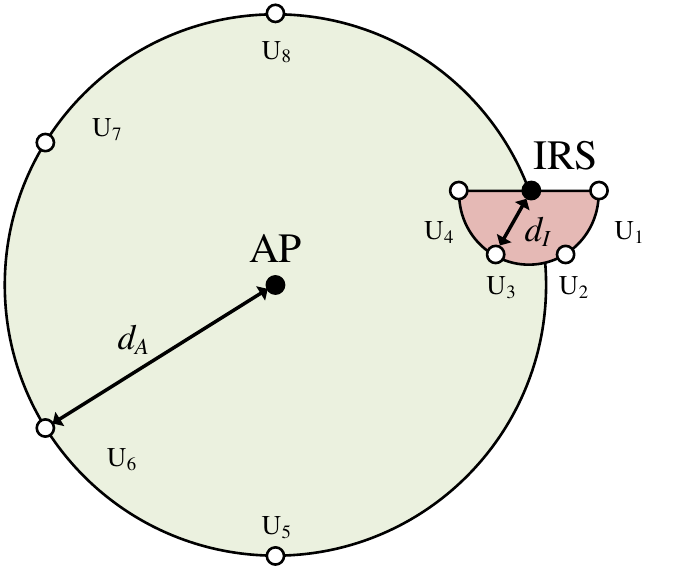}
\caption{Simulation setup of the multiuser system  (top view) where the deployment of the  AP and IRS is the same as that in Fig. \ref{simulation:setup:SU}.} \label{simulation:MU:setup}\vspace{-5mm}
\end{figure}

  \begin{figure}[ht]
\centering
\includegraphics[width=0.6\textwidth]{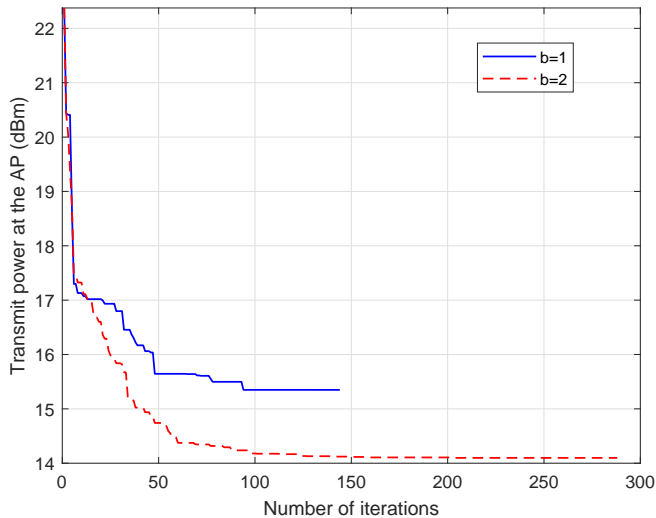}
\caption{Convergence behaviour of the ZF-based successive refinement algorithm for the multiuser case with $M=8, K=8, N=48$, and $\gamma=15$ dB. } \label{convergence} \vspace{-3mm}
\end{figure}

\subsection{Multiuser System}

Next, we consider a multiuser system with eight users, denoted by $U_k$'s, $k = 1,\cdots ,8$, and their locations are shown in Fig. \ref{simulation:MU:setup}. Specifically,  $U_k$'s, $k=1,2,3,4$, lie evenly on a half-circle centered at the reference element  of the IRS  with radius $d_I=2$ m and the rest users lie evenly on a half-circle centered at  the  reference antenna of the AP  with radius $d_A=50$ m. This setup can practically correspond to the case that the IRS is deployed at the cell-edge to cover an area with a high density of users (e.g., a hot spot scenario).  The channel parameters are set as $\alpha_{\rm AI}=2.2$, $\alpha_{\rm Iu}=2.8$, $\beta_{\rm AI}= \infty$, and $\beta_{\rm Iu}= 0$. First, we  show the convergence behaviour of the proposed successive refinement algorithm in Section IV-B with $M=8, K=8, N=48$, and $\gamma=15$ dB. As shown in Fig. \ref{convergence}, it is observed that  this suboptimal algorithm converges more rapidly for the case of  $b=1$ as compared to that of   $b=2$, while their required complexities are much smaller than that of the optimal exhaustive search, i.e.,  $\mathcal{O}(2^{bN})$.

\begin{figure}[ht]
\centering
\subfigure[$M=4, N=8, K=2$]{\includegraphics[width=0.55\textwidth]{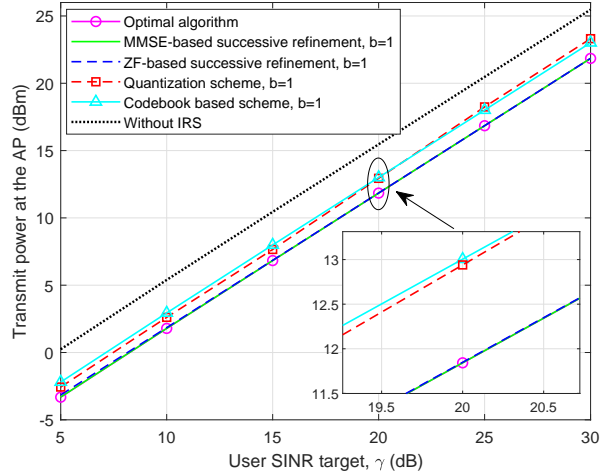} \label{transmit:pow:SINRa} } 
\subfigure[$M=6, N=32, K=4$]{\includegraphics[width=0.49\textwidth]{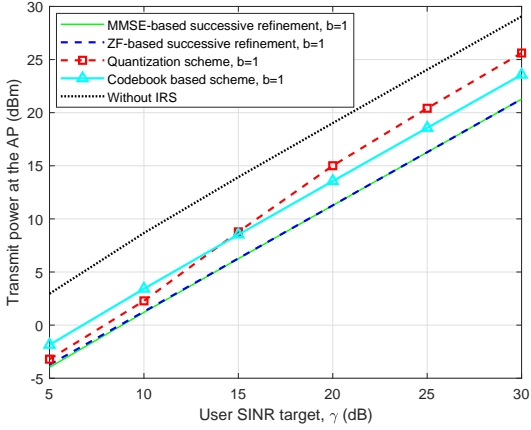}\label{transmit:pow:SINRb}} \!\! \!\! \!\! \!\!
\subfigure[$M=6, N=64, K=4$]{\includegraphics[width=0.49\textwidth]{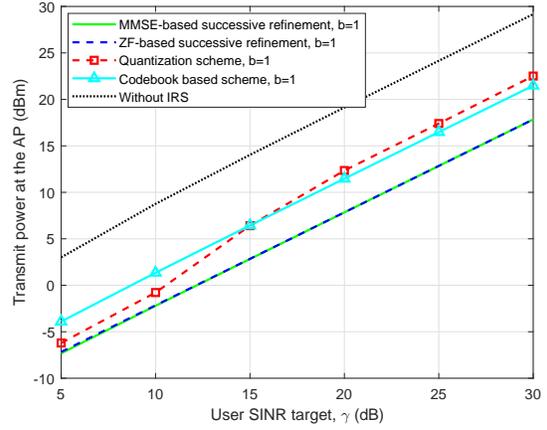}\label{transmit:pow:SINRc}} 
\caption{AP transmit power versus the user SINR target under different setups.  } \label{transmit:pow:SINR} \vspace{-3mm}
\end{figure}

\subsubsection{Performance Comparison with  Benchmark Schemes}
 In Fig. \ref{transmit:pow:SINR}, we plot the transmit power versus the user SINR target by setting $b=1$ under different system setups. We assume that $U_k$'s, $k = 1,2$, are active (need to be served) for $K=2$ and $U_k$'s , $k=1,2,3,4$ are active for $K=4$. Due to the high complexity of exhaustive search in the optimal algorithm, we consider it as a benchmark scheme only for a relatively small system size shown in Fig. \ref{transmit:pow:SINRa}, while for Figs. \ref{transmit:pow:SINRb} and \ref{transmit:pow:SINRb}, we propose an MMSE-based successive refinement algorithm as the benchmark. Specifically, in each iteration, we search all the possible values of $\theta_n$ over $\mathcal{F}$ by fixing $\theta_\ell$'s, $\forall \ell\neq n, \ell\in \mathcal{N}$, and for each value, we solve (P3) to obtain the MMSE precoder  as well as the corresponding AP transmit power. If (P3) is not feasible for a specific phase-shift value in $\mathcal{F}$, the required AP transmit power is set as positive infinity. Then, the phase-shift value that corresponds to the minimum AP transmit  power is chosen as the optimal $\theta_n$ in each  iteration. The above procedure is repeated until the fractional decrease of the objective value is less than the pre-defined threshold.  Since the MMSE precoder is the optimal solution to (P3), the transmit power of the MMSE-based successive refinement algorithm generally serves as a lower bound for that of the ZF-based   successive refinement algorithm.
We compare the following schemes. 1) Optimal algorithm in Section IV-A (for Fig. \ref{transmit:pow:SINRa} only); 2) MMSE precoding based successive refinement algorithm given above; 3) ZF precoding based successive refinement algorithm proposed  in Section IV-B;  4)  Quantization scheme: quantizing the continuous phase shifts obtained by using the iterative algorithm in \cite{JR:wu2018IRS} to their respective nearest values in $\mathcal{F}$; 5) Codebook based scheme as in the single-user case; 6)  Benchmark scheme without the IRS.  For schemes 5) and 6), the MMSE precoder  is applied at the AP.

From Figs. \ref{transmit:pow:SINRa}-\ref{transmit:pow:SINRc}, it is first observed  that all the algorithms with IRS achieve significant transmit power reduction at the AP as compared to the case without IRS, which demonstrates the effectiveness of IRS in the multiuser scenario.  Second, one can observe from Fig. \ref{transmit:pow:SINRa} that the proposed ZF-based successive refinement algorithm achieves near-optimal performance and outperforms both the quantization and codebook based schemes.  In addition, by comparing  Figs. \ref{transmit:pow:SINRa} and \ref{transmit:pow:SINRc}, it is observed that the performance gain of  the proposed ZF-based algorithm over benchmark schemes becomes more pronounced as the system size becomes larger.  This is expected since for the codebook based scheme, the possible combinations of phase-shift vectors grow exponentially as $N$ increases, while the codebook based scheme only  linearly increases the codebook size with $N$.  It is worth pointing out that although the quantization scheme suffers from small performance loss in the low SINR regime compared to the proposed  ZF-based algorithm, it performs even worse than the codebook based scheme in the high SINR regime, as shown in  Figs. \ref{transmit:pow:SINRa}-\ref{transmit:pow:SINRc}. This is because  the multiuser interference becomes severe when the user SINR target is high, and thus a coarse quantization from continuous phase shift values to discrete ones results in  signal mismatch not only  in desired signal combining  but also in interference cancellation.      Finally, from Figs. \ref{transmit:pow:SINRa}-\ref{transmit:pow:SINRc}, one can observe that the ZF-based algorithm performs almost the same as the MMSE-based algorithm for a wide range of SINR targets in all considered setups.  The reason behind such a phenomenon is that the IRS can effectively reduce the undesired  channel correlation among users  via  providing additional controllable multi-path signals to its nearby users.  

\begin{figure}[ht]
\centering
\includegraphics[width=0.6\textwidth]{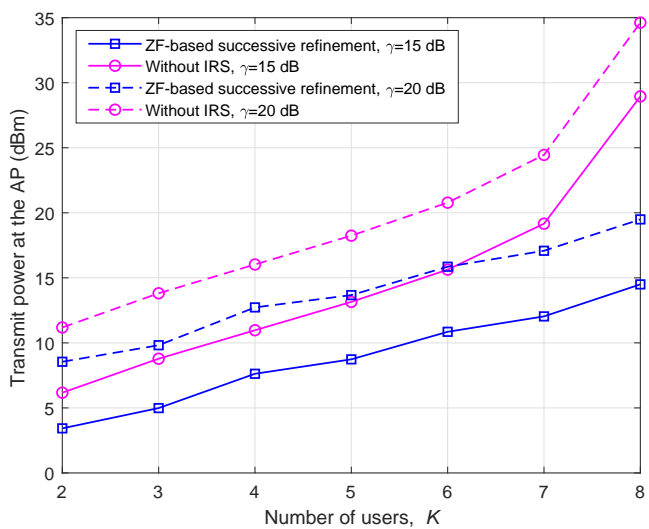}
\vspace{-0.4cm}
\caption{AP transmit power versus the number of users.} \label{simulation:K} \vspace{-5mm}
\end{figure}

\subsubsection{AP Transmit Power versus Number of Users}  In Fig. \ref{simulation:K}, we show the AP transmit power versus the number of users by setting $M=8$, $N=48$, and $b=1$. All other parameters are the same as those in Fig. \ref{transmit:pow:SINR}. In particular, we follow the user index order and successively add one user ($k=1,2,3,4$) near the IRS and then one user ($k=5,6,7,8$) far from the IRS to draw useful insights. Note that  $U_k$'s, $k=1,2,3,4$, located in the vicinity of the IRS, have similar path loss as the other  $U_k$'s, $k=5,6,7,8$, in the AP-user links.
From Fig. \ref{simulation:K}, it is first observed that adding a user near the IRS (e.g., adding $U_2$ corresponds to increasing $K$ from 2 to 3) requires less additional transmit power than that after  adding a user far from the IRS (e.g., adding $U_6$ corresponds to increasing  $K$ from 3 to 4), thanks to the passive beamforming gain provided  by the IRS.   More importantly, one can observe that  when the number of users approaches that of antennas at the AP, the transmit power in the case without IRS increases much faster than that in the case with IRS. This further demonstrates that the multiuser interference can be more effectively suppressed by applying the joint active and passive beamforming in the IRS-aided system.
Another important  implication of the above result is that the IRS has the capability of transforming a poorly-conditioned MIMO channel to a well-conditioned MIMO channel by adding more controllable multi-paths. For instance, for $K=M= 2$, the multiuser MIMO channel without IRS  has a rank approximately given by  ${\rm{rank}}(\bm{\HH}^H_d)=1$, if the two users have highly correlated AP-user channels; whereas by leveraging the IRS to actively contribute more signal paths, it is more likely to have  ${\rm{rank}}(\HH_{r}^{H}\ttheta\G + \HH_{d}^{H})=2$, thus helping reap the full spatial multiplexing gain in a multiuser MIMO system.

\subsubsection{IRS-aided Small MIMO versus Large MIMO without IRS}
\begin{figure}[ht]
\centering
\includegraphics[width=0.6\textwidth]{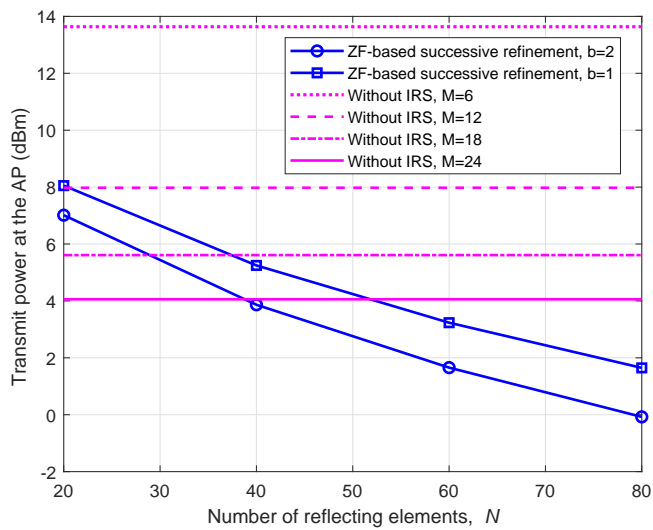} 
\vspace{-0.4cm}
\caption{AP transmit power versus the number of reflecting elements. } \label{simulation:MU:N} \vspace{-5mm}
\end{figure}
Due to the deployment of IRS, the number of transmit antennas at the AP can be reduced given the same AP transmit power and the users' SINR targets. This thus leads to a potentially more cost-effective solution for future wireless networks by using  small MIMO with low-cost IRS as compared to the traditional large (massive) MIMO without the IRS. To compare the performance of these two somewhat  opposite design paradigms, we show in Fig. \ref{simulation:MU:N} the AP transmit power versus the number of IRS elements with $M=6$, $\gamma=15$ dB, and $K=4$ (i.e., only the four users near the IRS are active). We consider the IRS with $b=1$ or $b=2$ as compared to a large MIMO system   without using IRS.
From  Fig. \ref{simulation:MU:N}, it is observed that for the AP transmit power of 4 dBm, we need to deploy 24 active antennas at the AP  in the case without IRS. In contrast, with  the same user SINR performance, we can alternatively use   a hybrid configuration by  deploying  only 6 active antennas at the AP together with either 52 1-bit or 38 2-bit passive reflecting elements at the IRS.  As a result, the associated RF power consumption and hardware cost for active antennas at the AP are significantly reduced over the case of large MIMO without IRS, thus providing a new cost-effective solution to achieve the same large MIMO performance gain.  Therefore,   the IRS-aided system provides more flexibility to trade-off between the number of active antennas ($M$) at the AP and that of  passive elements ($N$) at the IRS as well as their equipped phase shifters with different  levels ($2^b$), to optimally balance between the system performance and cost.

\section{Conclusions}

In this paper, we studied the beamforming optimization for IRS-aided wireless communications under practical discrete phase-shift constraints at the IRS.  Specifically, the continuous transmit precoder  at the AP and discrete phase shifts at the IRS  were jointly optimized to minimize the transmit power at the AP while meeting  the given user SINR targets. We  proposed both optimal and successive refinement based suboptimal  solutions for the single-user as well as multiuser cases. Furthermore, we  analyzed the performance loss of IRS due to discrete phase shifts  as compared to the ideal case with continuous phase shifts, when the number of reflecting elements becomes asymptotically large. Interestingly, it was shown that using IRS with even  1-bit phase shifters is still able to achieve the same asymptotic squared power gain as in the case with continuous phase shifts, subjected to only a constant power loss in dB. Simulation results showed that significant  transmit power saving can be  achieved by using IRS with discrete phase shifts as compared to the case without IRS, while the performance gains in terms of other metrics such as achievable rate and receive SINR can be similarly shown. In addition, it was revealed  that directly quantizing  the optimized  continuous phase shifts to obtain discrete phase shifts achieves near-optimal performance in the single-user case,  while its  performance degradation in the multiuser case  is non-negligible due to the severe co-channel interference. Finally, it was shown  that the ZF precoder based algorithm performs  almost as well as the MMSE precoder based algorithm, thanks to the multiuser channel rank improvement with the additional signal paths  provided by the IRS.

There are other  important issues that are not addressed in this paper yet, some of which are listed as follows to motivate future works.
\begin{itemize}
  \item Besides phase shifts optimization studied in this paper, the reflection amplitude of IRS's elements can be adjusted to further improve the system performance \cite{JR:wu2019IRSmaga}. However, the  joint optimization of reflection amplitude and phase shifts, in the form of either discrete or continuous values, is more challenging to solve.  In addition, it is unclear  whether the performance gain obtained by such joint phase-amplitude   optimization is sufficiently large to justify the increased hardware cost and algorithm complexity in practice, which needs further investigation.
  \item In practice, one AP may be assisted by multiple  IRSs while each IRS may be deployed to assist more than one APs. Although the local coverage of passive IRSs greatly simplifies the inter-IRS interference management if they are properly separated,  the joint user association, transmit precoding, and phase shifts optimization over multiple APs/IRSs is more involved than the single AP/IRS design problem considered in this paper, and thus is worthy of further investigation.
      \item In the multi-cell scenario, the joint deployment of the active APs and passive IRSs is also an important problem to investigate in future work. For example,  the densities of APs and IRSs as well as their locations can be jointly optimized to achieve the desired communication performance at minimum system cost.
\end{itemize}


\end{document}